\documentclass[11pt,letterpaper]{amsart}

\usepackage{header}

\usepackage{url}
\numberwithin{equation}{section}
\begin{document}
\title{What is Ballistic Transport? }
\author{David Damanik}
\address{Department of Mathematics\\
  Rice University\\
 Houston, TX, 77005}
\email[]{damanik@rice.edu }
\thanks{D.\ D.\ was supported in part by NSF grant DMS--2054752}
\author{Tal Malinovitch}
\address{Department of Mathematics\\
  Rice University\\
 Houston, TX, 77005}
  \email[]{tal.malinovitch@rice.edu}
\author{Giorgio Young }
\address{Department of Mathematics\\
  The University of Michigan\\
 Ann Arbor, MI 48109}
  \email[]{gfyoung@umich.edu}
  \thanks{G.Y.\ acknowledges the support of the National Science Foundation through grant DMS--2303363.}
\maketitle

\maketitle

\begin{abstract}
In this article, we review some notions of ballistic transport from the mathematics and physics literature, describe their basic interrelations, and contrast them with other commonly studied notions of wave packet spread. 
\end{abstract}

\section{Introduction}\label{intro}
One of the important questions when considering a quantum mechanical system is to understand the dynamics of states under the evolution of the system. A quantum mechanical system is described, in this context, by the Hamiltonian, by which we mean the operator defined by
\begin{align}\label{Hamiltonian}
    H=-\Delta+V=H_0+V,
\end{align}
where $\Delta$ is the Laplacian (acting on $\bbZ^d$ or $\bbR^d$) and $H_0$ its self-adjoint realization, while $V$ is a real-valued function acting by multiplication, and we fix a domain so that $H$ is a self-adjoint operator. In this survey, we will take $V$ to be a bounded function, though some of the results may apply in more general settings. Jacobi matrices are one of the particularly well-studied generalizations of \eqref{Hamiltonian} in the discrete one-dimensional setting; we only mention these sporadically throughout the text, although much of our discussion is valid for them as well.\par 

Throughout the paper, we will write $\calH$ to denote our Hilbert space, be it $L^2(\bbR^d)$ or $\ell^2(\bbZ^d)$. The system's dynamics are described by the time-dependent Schr\"odinger equation:
\begin{align}\label{Schrod}
    i\partial_t\psi=H\psi,\;\psi(x,0)=\psi_0.
\end{align}
Using the spectral theorem, one may define the propagator $e^{-itH}$ and denote by $\psi_t=e^{-itH}\psi_0$ the solution to the initial value problem \eqref{Schrod}.  \par
One way to study the dynamics of a state under this evolution is to compare aspects of the dynamics to those of the free propagator, $e^{-itH_0}$- which is well understood. For example, in certain settings, one may compare the dynamics of a state $\psi$ directly to the free evolution by finding some other state $\phi$ such that
\begin{align*}
    \|e^{-itH_0}\phi - e^{-itH}\psi \|\rightarrow 0,
\end{align*}
where $\|\cdot \|$ is the norm on $\calH$. This idea leads to scattering theory and the study of wave operators. \par
Though in many circumstances, this approach suffices, in others, the dynamics do not allow for direct comparisons to the free evolution of scattering theory. This is true even for fundamental physical examples, such as a periodic system, i.e. $H$ with a periodic potential $V$. Then, one seeks to compare different characteristics of the full dynamics to those of the free one. One important characteristic is how fast an initial state spreads in space. There are many ways to measure this property. Many of them may be viewed as making use of the following heuristic: if one takes an initial state that is localized in a ball of radius $R$ and has momentum $\sim P$ and applies the free evolution, after time $t$, the state will be roughly localized in a ball of size $R+Pt$. \par
In this survey, we will focus on one particularly natural way to measure this, and that is through the presence of ballistic transport. Roughly, this means that for $Q(t)$ the  Heisenberg evolved position operator, $\frac{1}{t}Q(t)\psi$, tends to a nonzero state asymptotically. This may be considered, in some cases, as the claim that this state has a non-zero asymptotic momentum. Given an initially localized state $\psi_0$, the quantity $\|Q(t)\psi_0\|^2$ measures the variance of the probability distribution of the particle's position at time $t$, which is the first moment that could be used to measure the outward spread of this distribution in time. There are different, non-equivalent ways to formulate this, and we will give precise definitions and explore the relations between these definitions in Section \ref{def}. \par
This notion of ballistic transport is part of a spectrum of transport, where two other distinguished phenomena are quantum diffusion and the absence of transport. As mentioned above, ballistic transport implies that the wave packet spreads linearly in time,  $x\sim t$. One can contrast this with the behavior of an eigenfunction, that the evolution operator simply rotates, so that the wave packet does not spread, i.e., $x\sim 1$. In this case, we say that the state exhibits absence of transport\footnote{Not to be confused with an absence of \textit{ballistic} transport, see Remark~\ref{absenceremark}}. In between ballistic and absence of transport, we have quantum diffusion, where $x\sim t^{\frac{1}{2}}$, which is the generally expected behavior in disordered systems in higher dimensions, see, for example, \cite{stolz2011introduction} for more details. With this spectrum in mind, we aim to clarify the definitions of $x\sim t$ in the literature, akin to the work of \cite{del1995localization}, entitled \textit{What is Localization?} in the context of dynamical localization, which implies absence of transport. \par

While our focus in this survey is on ballistic transport, we will briefly discuss the connection between this notion of wave packet spread and three other common ways that capture this phenomenon: spatial localization (or escape probabilities), scattering methods, and dispersion. By spatial localization, not to be confused with spectral localization, we mean the measurement of the main support of the function: considering statements on the change in $\|\chi_A e^{-itH}\psi\|$, where $\chi_A $ is the indicator of a subset $A$ of the ambient space. One of the origins of localization is in the famous RAGE theorem \cite{AmreinGeorgescu, Ruelle}, which tracks the probability of finding the particle outside of compact sets for a large time, in some sense, measuring the $0$'th moment or mass of the probability density function. There, a profound connection was proven between the time decay of this quantity inside and outside of arbitrary compact sets and the continuity of the spectral measure of the initial state. Roughly, it states that a pure point state is always localized in some sense, and continuous states, on average, escape any confined set. This leads to another, softer way of measuring the spread of the set by restricting it to a set of the form $B_{r+vt}$, the ball of radius $r+vt$ for some $v$. Section \ref{Rage} will explore the connection between this notion and ballistic transport. \par
Next, in Section \ref{scattering}, we will discuss the connection between some notions in scattering theory and ballistic transport. As mentioned above, scattering theory allows one to compare the asymptotics, in some circumstances, of the evolution under the full operator to that of the free operator. As we will discuss, some of the tools developed in scattering theory also allow one to conclude certain types of ballistic transport within their scope. \par
In significantly more general settings, one may try to measure dispersion, typically by finding a dispersive estimate. In some sense, these estimates aim to compare the evolution of the full operator to the evolution of the behavior of a known operator, for example, the free one, through different $L^p$ norms. The propagator $e^{-itH}$ is unitary and therefore preserves the $L^2$ norm. A dispersive estimate is usually an $L^1 \rightarrow L^\infty $ norm estimate of the propagator, which allows interpolation for any $1\leq p\leq \infty$ norm. Roughly speaking, the $L^1$ norm measures the spread of the function, while the $L^\infty $ measures the function's peak. If there is dispersion, as there is for the free evolution, since the $L^2 $ mass is preserved, one expects the peak of the wave packet to decay like the square root of the volume of the support of the function. So, by this heuristic, we may expect the following dispersive estimate,
\begin{align*}
    \|e^{-itH}\psi \|_\infty \leq Ct^{-\frac{d}{2}}\|\psi\|_1,
\end{align*}
as the volumes increases like $t^d$, where $d$ is the ambient dimension. Section \ref{disp} will explore the connection between this notion and ballistic transport. There are many important results in all of these fields, which this margin is too narrow to contain. So, we will focus only on these results in the context of ballistic transport.\par
This paper is arranged as follows:
In Section \ref{def}, we give the necessary notation and define the different notions of ballistic transport present in the literature, while Section \ref{Connections} explores the relations between these different notions.
In Section \ref{past}, we give an overview of some of the major results on ballistic transport, from general results in Section \ref{general} relying on the continuity of the spectral measure to a brief discussion of some results on the other types of transport in Section~\ref{otherTransport}. Finally, we give a brief overview of the transport results for periodic, quasi-periodic, and limit-periodic $V$ in Section \ref{semiperiodic}.
Section \ref{DiffNotion} explores the connection between ballistic transport and other notions of wave packet spread in terms of escape probabilities in Section \ref{Rage}, scattering theory in Section~\ref{scattering}, and dispersive estimates in Section \ref{disp}. Finally, the last section, Section \ref{physics}, will review the different use of the term ``ballistic transport" in the physics literature in Section \ref{BTphysics}, and then give an overview of some results in the physics literature that use notions closer to those discussed in the mathematical community in Section \ref{MathBTPhysics}.

\subsection*{Acknowledgement} We thank Adam Black, Jake Fillman, Svetlana Jitomirskaya, Ilya Kachkovskiy, Milivoje Luki\'c, and Hermann Schulz-Baldes for helpful discussions about this work. In particular, Open Questions \ref{finitegap} and \ref{SYClass} are due to Jake Fillman.

\section{Mathematical Definitions}\label{def}

Throughout this paper, we will consider a system described by the Hamiltonian defined by \eqref{Hamiltonian}, acting on $\calH=L^2(\mathcal X)$, where $\mathcal X$ denotes the ambient space, that is, $\mathcal X=\bbR^d$ or $\mathcal X=\bbZ^d$.  We will limit our scope to considering a Hamiltonian $H$ that is self-adjoint. For most of this manuscript, we will take $V$ to be bounded, unless otherwise specified, and then for $\mathcal X=\mathbb R^d$, the Kato-Rellich theorem yields that $H$ is self-adjoint on $H^2(\mathbb R^d)$, and for $\calX=\bbZ^d$ we have that $H$ is self-adjoint on all of $\calH$ as a bounded operator. \par 

We may now define the different notions of ballistic transport: Let $Q$ be the position operator on $\calH$ 
\begin{align*}
    &Q:\calH\rightarrow\calH^d\\
    &Q\psi=\vec{q}\,\psi,
\end{align*}
where $\vec{q}=(q_1,\ldots,q_d)$, with domain
\begin{align*}
    &D(Q)=\left\{ \psi\in\calH:\|Q\psi\|<\infty  \right\}.
\end{align*}
On a similarly natural domain, we will also define the $p$-th power of the position operator
\begin{align*}
    |Q|^{p}\psi=\|\vec{q}\|^p\psi. 
\end{align*}
We will also denote by $|Q|$ the operator of multiplication by the Euclidean norm of $\vec q$.\par
For an operator $A$, we denote the Heisenberg-evolved operator by 
\begin{align*}
A(t):=e^{itH}Ae^{-itH}.
\end{align*}

We will also need to consider the momentum operator, which we will denote by $\Xi$, defined as
\begin{align*}
    \Xi=-i[Q, H_0]
\end{align*}
with its maximal domain. So, the domain of $\Xi$ in the discrete setting is $\calH$ (as the momentum operator is bounded by $2$), and in the continuous setting, it will be $H^1$. \par 

We also define the transport exponents for $\psi $ in the appropriate domain,
\begin{align*}
    \beta^+_\psi(p)= \limsup\limits_{t\rightarrow\infty }\frac{\log \||Q|^{\frac{p}{2}}(t)\psi\|^2}{p\log t},\;&\beta^-_\psi(p)=\liminf\limits_{t\rightarrow\infty }\frac{\log \||Q|^{\frac{p}{2}}(t)\psi\|^2}{p\log t}.
\end{align*}
We note that we define the transport exponent for all $p > 0$, as is common, particularly in the discrete setting. Matching the continuum literature, our definition below takes $p=2$. However, the definitions extend naturally to different moments, see discussion below Definition \ref{BTdef}. We also note that a weaker form of transport exponents is often defined in terms of the Abel average
\begin{align*}
    \braket{\braket{Q_\psi ^p}}_A(T)=\frac{2}{T}\int\limits_0^\infty e^{-\frac{2t}{T}}\||Q|^{\frac{p}{2}}(t)\psi\|^2\, dt
\end{align*} 
for suitable states $\psi$. We also define the corresponding averaged transport exponents: 
   \begin{align*}
        \tilde{\beta}^+_\psi(2)= \limsup\limits_{t\rightarrow\infty }\frac{\log \braket{\braket{Q_\psi^2}}_A}{2\log t},& \quad \tilde{\beta}^-_\psi(2)=\liminf\limits_{t\rightarrow\infty }\frac{\log\braket{\braket{Q_\psi^2}}_A}{2\log t},
    \end{align*}
and we have that in applications, Abel averaged transport exponents often agree with those defined through the, perhaps more natural, Ces\`aro average
\begin{align*}
    \braket{\braket{Q_\psi ^p}}_C(T)=\frac{1}{T}\int\limits_0^T \||Q|^{\frac{p}{2}}(t)\psi\|^2\, dt
\end{align*} 
cf. \cite{DamanikLenzStolz}.
With these definitions, we can define the following local notions of ballistic transport for a state $\psi $:
\begin{definition}\label{BTdef}
    \begin{enumerate}
    \item We will say that a state $\psi$ exhibits \emph{strong ballistic transport} if it has a nonzero asymptotic velocity: $\frac{Q(t)}{t}\psi \rightarrow P_\psi $ for $P_\psi$ is not the zero state. 
    \item We say that a state $\psi $ exhibits \emph{ballistic transport in the norm-growth sense} if ballistic upper and lower bounds exist: there exist constants $c,C>0$ such that 
    \begin{align}\label{NGBT}
       ct^2\leq \||Q|(t)\psi\|^2\leq Ct^2 .
    \end{align}
    \item We say that $\psi$ exhibits \emph{ballistic transport in the exponent sense} if we have that $\beta^{\pm}_\psi(2)=1$.
    \item We say that $\psi$ exhibits \emph{ballistic transport in the Abel average norm-growth sense} if there are constants $c, C>0$ such that 
    \begin{align*}
        cT^2\leq \braket{\braket{Q_\psi^2}}_A (T) \leq CT^{2}.
    \end{align*}
    \item We say that $\psi $ exhibits \emph{ballistic transport in the Abel average exponent sense} if $\tilde{\beta}^\pm_\psi(2)=1$.
\end{enumerate}
Each of these notions can be modified in several ways:
\begin{itemize}
    \item Averaged ballistic transport: We can consider each of the notions above when replacing the Heisenberg evolved operator $Q(t)$ with the Ces\`aro averaged operator
    \begin{align*}
        \braket{Q}(T)=\frac{1}{T}\int\limits_0^T Q(t) \, dt .
    \end{align*}
    We will denote the average transport exponent by $\braket{\beta}^\pm_\psi(p)$. 
    \item Different moments: In all the definitions above, we can consider different moments, extending the definition naturally. In the case of strong ballistic transport, we will replace  $Q(t)/t$ with the $p$-th moment, $|Q|^p(t)/t^p$ for $p>0$. We will refer to this as strong ballistic transport for the $p$-th moment. 
\end{itemize}
\end{definition}
Here are several remarks regarding these definitions:
\begin{remark}
    For dimension $d\geq 2$, we note that strong ballistic transport means the convergence of a vector in $\calH^d$, whereas the notions that follow essentially require upper and lower estimates on scalar quantities defined by the norms $\||Q|e^{-itH}\psi\|$.
\end{remark}
\begin{remark}\label{absenceremark}
    As mentioned in the introduction, quantum diffusion of a state will correspond to $\beta_\psi^\pm (2)=\frac{1}{2}$ in the exponent sense, or convergence of $\frac{Q(t)}{\sqrt{t}}$ to a non-zero finite asymptotic operator, in a strong sense. Similarly, the absence of transport will correspond to $\beta_\psi^+(2)=0$ in the exponent sense or convergence of $Q(t)$ to an asymptotic non-zero finite operator in the strong sense. Note that the absence of transport is different than the absence of ballistic transport, which corresponds to $\beta_\psi^\pm (2)< 1$, or its strong counterpart.
\end{remark}
\begin{remark}
    In the discrete setting, it is very common to prove estimates for the state $\delta_0$, as a representative initially localized state. See, for example, the definitions in \cite{DFESO1}.
\end{remark}
\begin{remark}
    It may also prove convenient to define ballistic transport in terms of the difference $Q(t)-Q(0)$. The advantage of such an approach (that was taken in, for example, \cite{KachkovskiyGe}) is that it may avoid domain consideration, particularly in the discrete setting. In this survey, we will work with the definitions above, and we make domain considerations explicit. 
\end{remark}
In addition to these notions, we have another common definition of ballistic transport based on strong resolvent convergence:
\begin{definition}\label{resolvent}
    We will say that we have a \emph{strong resolvent ballistic transport} if there is a self-adjoint operator $P$, such that $\frac{Q(t)}{t}\rightarrow P$, in the strong resolvent sense, and that $\ker (P)\cap D(Q)=\{0\}$. 
\end{definition}
We note that once one has proven strong ballistic transport on a set $W$, the rule $\psi\to P_\psi$ defines a possibly unbounded symmetric linear operator with domain $W$, and it is then natural to ask if one can conclude that this operator is essentially self-adjoint, and if strong resolvent convergence holds. We give a criterion for this below.

We also note that while the notions described in Definition~\ref{BTdef} are all local, that is, they describe a state's quantum evolution and can be defined on a single state a priori (though that might not be an interesting result), whereas strong resolvent convergence is global in the sense that it says something about all the states in $D(Q)$ at once.

\section{Basic Results and Connection Between Different Notions of Ballistic Transport}\label{Connections}

In this section, we will cover some of the known results that hold in high generality as well as prove the simple implications between the different notions of ballistic transport defined above. We emphasize that in this section- due to the generality sought- we would not assume that the potential is bounded unless otherwise specified.  \par
We note that strong resolvent convergence does not immediately imply strong ballistic transport in the above sense. However, we give a necessary criterion for the implication to hold. \par
We show that the following picture holds (see Figure \ref{Illustration}):\par
\begin{figure}[h!]
\centering
    \begin{tikzpicture}[scale=1]
        \pgfmathsetmacro {\sca }{ 2}
        \pgfmathsetmacro {\shift }{0.5}
         \draw[black](-4,3*\sca) circle (0pt)node[anchor=south] {Strong Resolvent BT};
         \draw[black](0,3.5*\sca) circle (0pt)node[anchor=center] {$+$ common core};
         \draw[black](-5,2.25*\sca)  circle (0pt)node[anchor=center] {$+\ker (P)\cap D(Q)=\{0\}$};
         \draw[black](4,3*\sca) circle (0pt)node[anchor=south] {Strong BT};
         \draw[black](4,2*\sca) circle (0pt)node[anchor=south] {Norm Growth};
         \draw[black](0,\sca) circle (0pt)node[anchor=south] { BT Exponent};
         \draw[black](3,\sca) circle (0pt)node[anchor=south] {Abel Average BT};
         \draw[black](0,0) circle (0pt)node[anchor=south] {Abel Average BT Exponent};
         \draw[black,thick,->] (4,3*\sca) -- (4,2*\sca+\shift) ;
         \draw[black,thick,->] (4,2*\sca) -- (3,\sca+\shift) ;
         \draw[black,thick,->] (4,2*\sca) -- (0,\sca+\shift) ;
         \draw[black,thick,dashed, ->] (-4,3*\sca) -- (0,\sca+\shift) ;
        \draw[black,thick,->] (0,\sca) -- (0,0+\shift)  ;
        \draw[black,thick,->] (3,\sca) -- (0,0+\shift) ;
         \draw[black,thick,dashed, ->] (4-2*\shift,3*\sca+\shift/2+0.1) -- (-4+4*\shift,3*\sca+\shift/2+0.1) ;
    \end{tikzpicture}
    \caption{An illustration of which notion implies which notion in full generality, for a fixed moment. Dashed arrows represent implications that require additional assumptions, which are written next to the arrow. BT stands for ballistic transport. Note that the horizontal arrow passes from a condition on a state to one on the level of operators; proofs of this condition proceed by proving strong ballistic transport on a subset of states that forms a common core.} \label{Illustration}
\end{figure}
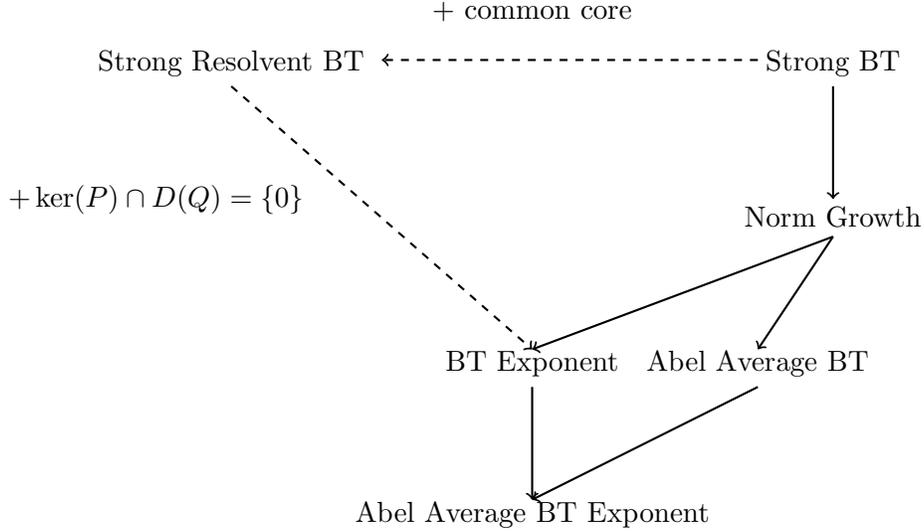
\begin{proposition}\label{Implications}
    Let $H$ be as above, then we have the following:
    \begin{enumerate}
        \item In the presence of ballistic upper bounds, strong resolvent convergence implies transport in the exponent sense. More precisely, if for all $\psi \in W$, for $W\subset D(Q^p)$ some subset, $\beta_\psi^\pm(p)\leq 1$, strong resolvent ballistic transport with $\ker(P)\cap W=\{ 0\}$ implies $\beta_\psi ^{\pm}(p)=1$ for all $\psi \in W$.\label{SRImpliesTE}
        \item Strong ballistic transport implies ballistic transport in the norm-growth sense.\label{StrongToNorm}
        \item Norm-growth ballistic transport implies ballistic transport in exponent sense.\label{NormToExp}
        \item Norm-growth ballistic transport implies ballistic transport in the norm-growth Abel average sense.
        \item Ballistic transport in the norm-growth Abel average sense implies ballistic transport in the Abel average exponent sense.
        \item Ballistic transport in the exponent sense implies ballistic transport in the Abel and Ces\`aro average exponent sense. 
        \item Any notion of ballistic transport implies the analogous notion for the averaged ballistic transport. 
    \end{enumerate}
\end{proposition}

\begin{remark}
    We note that in applications, due to \cite{RadinSimon}, $D(Q)\cap H^1(\mathbb R^n)\subseteq \bigcap_{t\geq 0}D(Q(t))$. Instead, in the discrete setting, boundedness of the momentum operator yields $D(Q(t))=D(Q)$.\par
    In addition, this result and its generalization in \cite{boutetdemonvel2023ballistic} implies ballistic upper bounds of all moments in many cases, as long as the potential is smooth enough. 
\end{remark}
\begin{proof}
\begin{enumerate}
   \item The following argument appears in \cite{damanik2015quantum}, but we add it here for completeness. Let $g_C(x)=(\min(|x|,C))^p$. Then, since $g_C$ is a bounded continuous function, \cite[Theorem VIII.20]{RSVol1} and  strong resolvent convergence imply
   \begin{align*}
    \lim_{t\to\infty} \left\langle \psi,g_C\left(\frac{Q(t)}{t}\right)\psi\right \rangle=
    \langle \psi,g_C(P)\psi\rangle.
   \end{align*}
    Since $g_C(x)$ is increasing in $C$ and $g_C(x)\to |x|^p$, we have, after taking $C\to\infty$,
    \begin{align*}
        \liminf_{t\to\infty} \left\langle \psi,\left(\frac{|Q(t)|}{t}\right )^p\psi\right \rangle\geq 
    \langle \psi,|P|^p\psi\rangle,
    \end{align*}
    and this final quantity is nonzero for any $\psi$ since $\ker(P)=\{0\}$. 
    \item Assume that we have ballistic transport in the strong sense, then it follows that
    \begin{align*}
        \lim \limits_{t\rightarrow\infty} \frac{1}{t^2}\||Q|(t)\psi\|^2=\|P_\psi\|^2.
    \end{align*}
    Let $\|P_\psi\|^2/2>\epsilon>0$ (as we have that $P_\psi\neq 0$). Then there is $R>0$ such that for all $t>R$, we have 
    \begin{align*}
        t^2(\|P_\psi\|^2-\epsilon)\leq \|Q(t)\psi\|^2\leq t^2(\|P_\psi\|^2+\epsilon),
    \end{align*}    
    so that by inflating and contracting the constants in the upper and lower bounds respectively, the result follows.
    \item Assume that we have ballistic transport in the norm-growth sense. Taking $\log$'s on both sides of \eqref{NGBT}, we have  
    \begin{align*}
       \log(c)+2\log(t)\leq \log\||Q|(t)\psi \|^2\leq \log(C)+2\log(t).
    \end{align*}
    Dividing by $2\log(t)$ and taking $t\to\infty$, we have $\beta^\pm_\psi(2)=1$, as needed. 
    \item Assume that we have norm-growth ballistic transport. Multiplying by $\frac{2e^{-2t/T}}{T}$ on both sides of \eqref{NGBT} and integrating by parts, we find
    \begin{align*}
        \frac{cT^2}{2}\leq \frac2T\int_0^\infty e^{-2t/T}\||Q|(t)\psi\|^2dt \leq \frac{CT^2}{2},
        \end{align*}
        so that by inflating and contracting the constants in the upper and lower bounds, respectively, the result follows.
    
    \item This comes from a similar argument as in the proof of (\ref{StrongToNorm}). 
    \item This result is fairly general, see \cite[Theorem 2.20]{damanik2010general}.
    \item We will show it only for $\psi$ with strong ballistic transport. A similar proof will work for the other notions as well. All proofs are based on the fact that 
    \begin{align*}
        \braket{Q}-P=\frac{1}{T}\int_0^TQ(t)-P\, dt,
    \end{align*}
      Let $\epsilon>0$ be given. Then there is some $t_0>0$ such that for all $t>t_0$,
    \begin{align*}
        \Big\|\Big(\frac{Q(t)}{t}-P\Big)\psi\Big\|<\epsilon.
    \end{align*}
    Then we have that 
    \begin{align*}
        \Big\|\Big(\frac{\braket{Q}(T)}{T} -P\Big)\psi\Big\|&\leq \frac{1}{T^2}\int\limits_0^T\|(Q(t)-P)\psi\|\, dt \\
        &=\frac{1}{T^2}\int\limits_0^{t_0}\|(Q(t)-P)\psi\|\, dt+\frac{1}{T^2}\int\limits_{t_0}^T\|(Q(t)-P)\psi\|\, dt\\
        &\leq \frac{t_0}{T^2}\sup\limits_{0\leq t\leq t_0} \|(Q(t)-P)\psi\|+\frac{\epsilon(T-t_0)}{T^2}\rightarrow 0
    \end{align*}
    as $T\rightarrow\infty $, as needed. A similar idea will yield the other implications. 
\end{enumerate}
\end{proof}

\begin{remark}
    We note that all implications of Proposition \ref{Implications} (except \ref{SRImpliesTE}) are written for the second moment but generalize for any moment as seen from the proof. 
\end{remark}

In addition, we have the following conditional implication:

\begin{proposition}
     Let $H$ be as above, let $W\subset D(Q)$, and let $Q(t)$ exhibit strong resolvent ballistic transport: $\frac{Q(t)}{t}\to P$. Then, for 
     \begin{align*}
         \psi\in \bigcap_{t\geq 0}D(Q(t))\cap D(P),
     \end{align*}
     $\frac{Q(t)}{t}\psi\to P\psi$ if and only if 
     \begin{align}
     \label{condition}\lim\sup_{t\to\infty}\left\|\frac{Q(t)}{t}\psi\right\|\leq \|P\psi\|.
     \end{align}
     In other words, the set \begin{align*}
     W_{max}:= \left\{\psi\in \bigcap_{t\geq 0}D(Q(t))\cap D(P): \limsup_{t\to\infty}\left\|\frac{Q(t)}{t}\psi\right\|\leq \|P\psi\|\right\}
     \end{align*}
     is maximal in the sense that for any set $W\subseteq \cap_{t\geq 0}D(Q(t))\cap D(P)$ on which there is strong ballistic transport, $W\subseteq W_{max}$.
\end{proposition}
\begin{proof}
    Fix $\phi,\psi\in \cap_{t\geq 0}D(Q(t))\cap D(P)$. Then, for $z\in \mathbb C\setminus \mathbb R$, we may write $\phi=(P-\overline{z})^{-1}\tilde \phi$ and $\psi=(\frac{1}{t}Q(t)-z)^{-1}\tilde\psi$. Then, by the second resolvent identity, 
   \begin{align*}
        \langle \phi,(Q(t)/t-P)\psi\rangle &=\langle \tilde \phi,(P-z)^{-1}(Q(t)/t-P)(\frac{1}{t}Q(t)-z)^{-1}\tilde \psi\rangle\\
        &=\langle \tilde \phi,((P-z)^{-1}-(\frac{1}{t}Q(t)-z)^{-1})\tilde \psi\rangle\to 0
   \end{align*}
    as $t\to \infty$ by strong resolvent convergence. Thus, since 
    \begin{align*}
        \left\|\left(\frac{Q(t)}{t}-P\right)\psi\right\|^2=\left\|\frac{Q(t)}{t}\psi\right\|^2-2\Re(\langle \frac{Q(t)}{t}\psi,P\psi \rangle)+\|P\psi\|^2
    \end{align*}
    we have
    \begin{align*}
    \limsup_{t\to\infty}\left\|(\frac{Q(t)}{t}-P)\psi\right\|^2=\limsup_{t\to\infty}\left \|\frac{Q(t)}{t}\psi\right\|^2-\|P\psi\|^2,
    \end{align*}
    which vanishes if and only if $\limsup_{t\to\infty}\|Q(t)/t\psi\|\leq \|P\psi\|$.
\end{proof}

In addition, we have the following theorem, sometimes used to find ballistic lower bounds on all moments in the discrete setting \cite{damanik2015quantum,Fillman}:
\begin{theorem}[\cite{RSVol1} Theorem VIII.25]
    Let $H$ be such that that $D $ is a common core for all $\frac{Q(t)}{t}$, and assume that every $\psi \in D$ exhibits ballistic transport in the strong sense. Then $H$ has ballistic transport in the strong resolvent sense. 
\end{theorem}
Naturally, the picture described above (and can be seen in Figure \ref{Illustration}) is missing the conditions that will allow us to ``go up the ladder," that is, get the reverse implications:
\begin{question}
    Under what condition on $H$ do we have the following implication:
    \begin{enumerate}
        \item If a state exhibits ballistic transport in the Abel average sense, then it exhibits ballistic transport in the strong sense.
        \item If a state exhibits ballistic transport in the exponent sense, then it exhibits ballistic transport in the strong sense.
        \item If a state exhibits ballistic transport in the Abel average exponent sense, then it exhibits ballistic transport in the exponent or Abel average sense.
    \end{enumerate}
\end{question}
In order to show the relation of any given notion for the different moments, we first need to introduce one of the fundamental  results of ballistic transport --- the Radin-Simon Theorem \cite{RadinSimon}.\par
This theorem gives a very general ballistic upper bound and shows the preservation of the domain $D(Q)$ under the propagator $e^{-itH}$ on $L^2(\mathbb R^d)$ under very general assumptions. 
\begin{theorem}[Radin-Simon]\label{Radin-Simon}
    Let $V$ be $H_0$-form bounded with a relative bound less than $1$. We define the space 
    \begin{align*}
        S_1=\{\psi \in \calH : Q\psi \in \calH, \Xi\psi \in \calH\},
    \end{align*}
    equipped with the norm
    \begin{align*}
        \|\psi\|_{S_1}=\sqrt{\|\psi\|^2+\|\frac{1}{2}\Xi\psi\|^2+\|Q\psi\|^2}.
    \end{align*}
    Then for $\psi \in D(Q)$, $e^{-itH}\psi \in D(Q)$, and we have that there is some $C>0$ such that for all $t\in \bbR$,
    \begin{align*}
        \|Qe^{-itH}\psi\|\leq C(1+|t|)\|\psi\|_{S_1}.
    \end{align*}
\end{theorem}
\begin{remark}
    We note that in the continuous case, we can write
    \begin{align*}
        \|\psi\|_{S_1}=\sqrt{\|\psi\|_{H_1}^2+\|Q\psi\|^2}.
    \end{align*}
    In the discrete setting, $\Xi$ is bounded so that 
    \begin{align*}
        \|\psi\|_{S_1} \leq \sqrt{3}\sqrt{\|\psi\|^2+\|Q\psi\|^2}.
    \end{align*}
\end{remark}
This result was later extended to higher moments, given that the potential is smooth enough; see \cite{boutetdemonvel2023ballistic}. Using the natural bound from Jensen's inequality, we get that 
\begin{align*}
    |Q|^p(t)\geq |Q|^{\tilde{p}}(t)
\end{align*}
if $p>\tilde{p}$. This yields the upper bound for all lower moments.\par  
Using this result, one can use the scattering method to show that most of the notions of ballistic transport are, in fact, stable under fast-decaying potentials. This is a direct application of the following lemma:
\begin{lemma}
    Let $H_0$ be a self adjoint operator acting on $\calH\in \{L^2(\bbR^d),\ell^2(\bbZ^d)\}$. Let $V$ be a bounded multiplication operator such that for some $\psi \in \calH$, we have
    \begin{align}\label{S1AssumDisc}
        &\int\limits_0^\infty (1+t)\|Ve^{-itH_0}\psi\|_{S_1} \, dt <\infty.
    \end{align}

    Then we have that for $H=H_0+V$, 
    \begin{align*}
        \lim\limits_{t\rightarrow \infty}\|(Qe^{-itH}\Omega-Qe^{-itH_0})\psi \|=0,
    \end{align*}
    where $\Omega=\slim\limits_{t\rightarrow\infty } e^{itH}e^{-itH_0}$ is the wave operator (for further discussion about the wave operator, see \ref{scattering}).
\end{lemma}
The proof of this theorem is the first part of the proof of Proposition B.3 in \cite{black2023directional}.\par
\begin{remark}
    Suppose that the potential $V$ is such that
    \begin{align*}
        &\|(1+R)Q\chi_{|x|>R}V\|_{\infty }\in L^1([0,\infty),dr).
    \end{align*}
    Then the above lemma applies, and in particular a fast decaying potential $V, \Xi V\sim x^{-2-\epsilon}$ is covered by this lemma.
\end{remark}
As an immediate corollary, we have
\begin{corollary}
    For $H_0, V, H$ as above, if $\psi $ exhibits ballistic transport for the second moment under any definition other than strong or strong resolvent, then $\Omega\psi $ exhibits the same notion.
\end{corollary}
\begin{proof}
    This is immediate, as all the notions rely on $\|Qe^{-itH}\Omega\psi\|$, and by the above, in the limit, it is asymptotically equal to $\|Qe^{-itH_0}\psi\|$.
\end{proof}

\begin{remark}
    We note that this is not trivially enough for strong ballistic transport, as one gets only that 
    \begin{align*}
        \lim\limits_{t\rightarrow \infty}\|(e^{itH}Qe^{-itH}\Omega-e^{itH}Qe^{-itH_0})\psi \|=0.
    \end{align*}
    In order to conclude strong ballistic transport, one needs to show that 
    \begin{align*}
        \lim\limits_{t\rightarrow \infty}\|e^{itH}Qe^{-itH_0}\Omega \psi-e^{itH_0}Qe^{-itH_0}\psi \|=0,
    \end{align*}
    which requires further assumptions.
\end{remark}

As explained above, one would expect that the free Hamiltonian $H_0=-\Delta $ defined on $L^2(\bbR^d)$ or $\ell^2(\bbZ^d)$ will have all types of ballistic transport, and indeed this is what we have:
\begin{theorem}
    For  $\psi \in Dom(Q)\cap Dom(\Xi)$ we have that for $H=H_0$,
    \begin{align*}
        \frac{Q(t)}{t}\psi \rightarrow \Xi\psi. 
    \end{align*}
     In particular, 
    \begin{align*}
    \frac{Q(t)}{t} \rightarrow \Xi 
    \end{align*}
    in the strong resolvent sense.
\end{theorem}
\begin{proof}
    
    We offer two proofs; both proceed by finding an asymptotic velocity on a common core of $Q(t)$ and $\Xi$. The second proof avoids treating the continuum and discrete setting separately, although the idea is identical in both.

    In the continuum, taking Fourier transform, denote by $\calF$, and denoting the momentum variable by $\xi$, we have 
    \begin{align*}
        \mathcal{F}\left(\frac{1}{t}e^{itH_0}Qe^{-itH_0}\psi\right)&=\frac1te^{it\|\xi\|^2}(2te^{-it\xi^2}\xi\hat{\psi}-e^{-it\|\xi\|^2}i\nabla_{\xi} \hat{\psi})\\
        &=2\xi\hat{\psi}+O_{L^2(\bbR^d)}(1/t).
    \end{align*}
    Inverting the Fourier transform, we see
    \begin{align*}
    \lim_{t\to\infty}\frac{1}{t}e^{itH_0}Qe^{-itH_0}\psi=-2i \nabla \psi,
    \end{align*}
    where $\nabla \psi\ne 0$ for nonzero $\psi\in \Dom(\Xi)=H^1(\mathbb R^d)$. As $\mathcal S(\mathbb R^d)\subset D(Q)\cap H^1(\mathbb R^d)$, we have shown convergence on a common core and may conclude strong resolvent convergence. 
    \par
    In the discrete setting, we take the discrete Fourier transform, also denoted $\calF$, to find 
     \begin{align*}
        \mathcal{F}\left(\frac{1}{t}e^{itH_0}Qe^{-itH_0}\psi\right)&=\frac1te^{it\sum\limits_{\ell=1}^d2\cos(\xi_\ell)}\left(-2\sin(\xi_j)te^{-it\sum\limits_{\ell=1}^d2\cos(\xi_\ell)}\hat{\psi}-e^{-it\sum\limits_{\ell=1}^d2\cos(\xi_\ell)}i\partial_{\xi_j} \hat{\psi}\right)\\
        &=-2\sin(\xi_j)\hat{\psi}+O_{\ell^2(\bbZ^d)}(1/t),
    \end{align*}
    and we may conclude as before.\par 
    
    Alternatively, we may start by noting that 
    \begin{align*}
        &[Q,H_0]=i\Xi,\\
        &[\Xi,H_0]=0.
    \end{align*}
    So we have that  (by a simple commutator identity)
    \begin{align*}
        -i[Q,H_0^m]=\sum_{n=0}^{m-1}H_0^n(-i[Q,H_0])H_0^{m-1-n}=\sum_{n=0}^{m-1}H_0^n\Xi H_0^{m-1-n}=mH_0^{m-1}\Xi. 
    \end{align*}
    We have
    \begin{align*}
        [Q,e^{-itH_0}]&=[Q,\sum_{n=0}^\infty\frac{(-itH_0)^n}{n!} ]=\sum_{n=0}^\infty\frac{(-it)^n}{n!}[Q,H_0^n]=\sum_{n=0}^\infty\frac{(-it)^n}{n!}nH_0^{n-1}i\Xi\\
        &=t\Xi\sum_{n=1}^\infty\frac{(-itH_0)^{n-1}}{(n-1)!}=t\Xi e^{-itH_0}.
    \end{align*}
    So we see
    \begin{align*}
        Q(t)=e^{itH_0}Qe^{-itH_0}=e^{itH_0}(e^{-itH_0}Q+e^{-itH_0}t\Xi)=Q+t\Xi.
    \end{align*}
    In particular, this implies that $D(Q)\cap D(\Xi)=D(Q(t))$. 
    Thus,
    \begin{align*}
        \frac{Q(t)}{t}\psi=\frac{Q}{t}\psi+\Xi\psi\rightarrow \Xi\psi
    \end{align*}
     for all $\psi \in D(Q)\cap D(\Xi)$. Since we have convergence of $Q(t)/t$ on its domain, we may conclude convergence in the strong resolvent sense. 
\end{proof}

\section{Past Results}\label{past}

Over the years, there have been many important results proving different notions of ballistic transport in different types of systems. \par
In this section, at first, we will survey these results that hold in great generality and provide the common base that other results rely on. Then we will review some results related to other types of transport. Finally, we will look more specifically at periodic and almost-periodic potentials, for which many results have been proven.  \par
As mentioned above, in this survey, we will focus on the results concerning the second moment, as this is the one most commonly used. So, unless stated otherwise, the result holds for the second moment.

\subsection{General Results}\label{general}

The most general results on transport typically rely on spectral decompositions of the Hamiltonian generating the time dynamics. In some sense, they may be viewed as quantitative descendants of the RAGE theorem \cite{AmreinGeorgescu, Ruelle} discussed in more detail in Section \ref{Rage}. One of the earlier results in this field was Simon's result about the absence of ballistic transport for operators with pure point spectrum. This result was the basis for the connection between spectral types and ballistic transport:
\begin{theorem}[Theorems 1.2 and 3.1 in \cite{simon1990absence}]
    Let $V$ be a multiplication operator such that $H$ is self-adjoint on $\calH$, and assume that $H$ has only pure point spectrum. Then for all $\varphi \in D(Q)\cap D(H)$, we have 
    \begin{align*}
        \lim\limits_{t\rightarrow \infty }\left\|\frac{Q}{t}e^{-itH}\varphi\right\|=0.
    \end{align*}
\end{theorem}
This result was recently generalized for operators with additional spectral types:
\begin{theorem}[Theorem A.1 in \cite{black2023directional}]\label{Absance}
    Let $H=-\Delta +V$ with $V$ relatively bounded. Let $\psi \in \calH_{\mathrm{pp}}\cap D(Q)\cap D(H)$. Then we have that
    \begin{align*}
        \lim_{t\rightarrow\infty} \left\|\frac{Q}{t}e^{-itH}\psi\right\|= 0.
    \end{align*}
\end{theorem}
In Appendix 2 of \cite{del1996operators}, the authors have shown that this result is sharp in the following sense:
\begin{theorem}
    There is a bounded potential function $V$ defined on $\bbZ_+$, such that the Hamiltonian
    \begin{align*}
        H=-\Delta+V
    \end{align*}
    acting on $\ell^2(\bbZ_+)$ with Dirichlet boundary conditions has a complete set of normalized eigenfunctions, and we have that 
    \begin{align*}
        \limsup\limits_{t\rightarrow\infty } \left\|\frac{Q(t)\sqrt{\ln(t)}}{t}\delta_0 \right\| = \infty. 
    \end{align*}
\end{theorem}
In particular, this theorem says that $\delta_0$ spreads at a rate $\sim \frac{t}{\sqrt{\ln(t)}}$. This result is extended to operators on $\ell^2(\mathbb Z)$ in \cite{DFESO2}, in the following strengthened form:
\begin{theorem}
Let $f:\mathbb R_+\to \mathbb R^+$ be monotone with $\lim_{t\to\infty}f(t)=\infty$. Then, there is a bounded potential function $V:\bbZ\to \mathbb R$, such that the Hamiltonian
    \begin{align*}
        H=-\Delta+V
    \end{align*}
    has a complete set of normalized exponentially decaying eigenfunctions and 
    \begin{align*}
        \limsup\limits_{t\rightarrow\infty } \left\| \frac{Q(t)f(t)}{t}\delta_0 \right\| = \infty.
    \end{align*}\end{theorem}
These results tell us that pure point states exhibit an absence of ballistic transport, but could have transport arbitrarily close to ballistic transport.\par
An extension of this result would be
\begin{question}\label{QuestionAbsance}
    Does Theorem \ref{Absance} hold for all moments?
\end{question}
One may ask whether the converse holds: If a state exhibits an absence of ballistic transport for each $p$-th moment, that is, $\psi$ is such that 
\begin{align*}
    \lim_{t\to\infty}\frac{1}{t^p}\||Q|^p(t)\psi\|=0
\end{align*}
for all $p>0$, does that imply that it is a pure point state? In $d\geq 3$ dimensions, the answer is no. An example is given in \cite{bellissard2000subdiffusive} where an a.c. state with subdiffusive Ces\`aro-averaged moment (that is $\braket{\beta}^\pm (2)<\frac{1}{2}$) transport was built. In fact, there were also examples of absolutely continuous states with $\beta^+(2)=0$, see \cite{vidal1999spectrum} for more details.   \par
Furthermore, in \cite{damanik2008quantum}, the authors showed that following
\begin{theorem}[\cite{damanik2008quantum}, Theorem 2]\label{Fibo}
    Let $H$ be given by
    \begin{align*}
        H=-\Delta+V
    \end{align*}
    acting on $\ell^2(\bbZ)$, where
    \begin{align*}
        V(n)=\lambda \chi_{1-\phi,1)}(n\phi+\theta \!\!\! \mod 1)
    \end{align*}
with $\lambda \geq 8$, $\phi=\frac{\sqrt{5}-1}{2}$ the inverse of the golden ratio, and $\theta \in [0,1)$.\par
    Denote
    \begin{align*}
        S(\lambda)=\frac{\lambda-4+\sqrt{(\lambda-4)^2-12}}{2}.
    \end{align*}
    Then we have that 
    \begin{align*}
         \beta_{\delta_0}^+(p)\leq \frac{2\log (1+\phi)}{\log (S(\lambda))}
    \end{align*}
    for all $p>0$. 
\end{theorem}
In particular, this theorem tells us that there are states (which, in fact, belong to the singular continuous subspace) with $\beta(2) < 1$. Choosing $\lambda$ carefully allows for an arbitrary slow spread. So, in particular, this implies the existence of non pure point states that do not exhibit ballistic transport for any moment. \par
A natural follow-up question: is there some stronger spectral condition that will force ballistic transport? The answer is yes. This result appeared in \cite{combes1992some}:
\begin{proposition}[\cite{combes1992some}, Proposition 1.2]
    Let $H$ be a Hamiltonian acting on $L^2(\bbR^d)$ and let $I\subset \sigma(H) $ be an interval on which there is a strict Mourre estimate, that is, for $A=[Q^2,H]$,
    \begin{align*}
        \exists \alpha>0, \chi_I(H)[H,A]\chi_I(H)\geq \alpha \chi_I(H).
    \end{align*}
    Then for $\psi \in \chi_I(H)\calH\cap D(Q)\cap H^1(\bbR^d)$, 
    \begin{align*}
        \|Q(t)\psi\|\geq C\alpha t
    \end{align*}
    for some constant $C>0$. This, combined with the Radin-Simon ballistic upper bound, implies ballistic transport in the strong sense.
\end{proposition}
\begin{remark}
    A strict Mourre estimate is known to hold in many cases; in particular, in \cite{derezinski1997scattering}  it is shown that for rapidly decaying potentials, such estimates hold. 
\end{remark}
\begin{proof}
      In \cite{combes1992some}, the proof is given in terms of 
    \begin{align*}
        \lim_{t \to \infty} \Big\|\frac{Q(t)}{t^{\frac{1}{2}}}\psi\Big\|=\infty. 
    \end{align*}
    But in fact, the proof allows for a much more general statement, and so we include the proof for the sake of completeness. This is a direct propagation estimate. Let $\psi \in \chi_I(H)\calH\cap D(Q)\cap D(H)$. Here, we will use the following identity, which holds for any operator $B$ and any $\varphi \in D(B)\cap D(H)\cap D([B, H])$:
    \begin{align*}
        \braket{\varphi ,B(t)\varphi }=\braket{\varphi ,B\varphi }+\int\limits_0^t \braket{\varphi,[B,H]\varphi}.
    \end{align*}
    Applying this identity twice yields the following:
    \begin{align*}
        &\|Q(t)\psi\|^2=\braket{\psi,Q^2(t)\psi}= \braket{\psi,Q^2\psi }+\int\limits_0^t\braket{\psi, i[H,Q^2](s)\psi }\, ds\\
        & = \braket{\psi,Q^2\psi }-\int\limits_0^t\braket{\psi, iA(s)\psi }\, ds \\
        & = \braket{\psi,Q^2\psi }-\int\limits_0^t[\braket{\psi, iA\psi }+\int\limits_0^s\braket{\psi, i[H,iA](\sigma)\psi }\, d\sigma]\, ds\\
        &=\braket{\psi,Q^2\psi }-t\braket{\psi,iA\psi}+\int\limits_0^t\int\limits_0^s\braket{\psi_\sigma, [H,A]\psi_\sigma }\, d\sigma\, ds\\
        &\geq \braket{\psi,Q^2\psi }-t\braket{\psi,iA\psi}+\int\limits_0^t\int\limits_0^s\alpha \braket{\psi_\sigma, \psi_\sigma }\, d\sigma\, ds\\
        &=\braket{\psi,Q^2\psi }-t\braket{\psi,iA\psi}+\alpha \|\psi\|^2\frac{t^2}{2}.
    \end{align*}
    Since $\braket{\psi,Q^2\psi },\braket{\psi,iA\psi}$ are bounded, we get that
    \begin{align*}
         \|Q(t)\psi\|^2\geq C\alpha \|\psi\|^2\frac{t^2}{2}
    \end{align*}
    for some $0<C$, as needed.
\end{proof}
The Mourre estimate mentioned above is a common estimate used in scattering theory (see \cite{derezinski1997scattering} for example for more details), and in particular, it implies purely absolutely continuous spectrum in this spectral interval. So, a ``strongly" purely absolutely continuous interval, in the sense of a Mourre estimate, implies ballistic transport.

A related question is whether the convergence of all moments and knowing that there is no kernel to the limiting operator excludes the existence of pure point states. One immediate implication of Theorem \ref{Absance} is the following:
\begin{corollary}
    Let $H$ be such that every nonzero $\psi\in D(Q)\cap D(H)$  exhibits ballistic transport in the strong sense. Then we have 
    \begin{align*}
        \calH_{pp}\subset \{0\} \cup (\calH\setminus D(Q)).
    \end{align*}
\end{corollary}
\begin{proof}
    Let $\psi \in \calH_{pp}\cap D(Q)$. By \cite[Theorem A.1]{black2023directional}, we have that $\|\frac{Q(t)}{t}\psi\|\rightarrow 0$, and hence, by the strong ballistic transport assumption, we must have $\psi=0$, as claimed.
\end{proof}
In one dimension, by Combes-Thomas estimates, this quickly implies $\sigma_d(H)=\emptyset$. However, this does not exclude the existence of eigenfunctions in $\calH\setminus D(Q)$; following the construction of a Wigner-Von Neumann type potential in \cite[XIII.13 ]{RSVol4}, one can construct an example of a potential with an eigenfunction that has decay of $\psi(x)\sim x^{-\frac{3}{4}}$ (on $\bbR$). In particular this implies that $ (\calH\setminus D(Q))\cap \calH_{pp}\neq \emptyset$. For more information about the construction of eigenfunctions of Wigner-Von Neumann type potentials with precise power law decay asymptotics, see \cite{LukicWignerVonNeumann,WVNNaboko,WVNSimon}.
This leads to an open question:
\begin{question}
     Let $H$ be such that every nonzero $\psi\in \bigcup_{p>0}D(|Q|^p)\cap D(H)$ exhibits ballistic transport in the strong sense for all the defined moments. Does it follow that 
    \begin{align*}
        \calH_{pp}= \{0\}?
    \end{align*}
\end{question}
\begin{remark}
    Given Theorem 4.4, a possible path towards answering this question (with a no) is clear: first, one needs to prove open question \ref{QuestionAbsance} and then show the existence of slowly decaying eigenfunctions (which will not be in $D(|Q|^p)$ for any $p>0$). 
\end{remark}

We now turn to the presence of transport depending only on the spectrum of the underlying operator. The works of Guarneri, Combes, and Last \cite{Combes, Guarneri1, Guarneri2, Last} proved lower bounds on the Ces\`aro averaged second moment for states with nonzero $\alpha$-continuous part in the Rogers-Taylor decomposition of their spectral measures. In particular, their results imply:
\begin{theorem}\cite[Theorems 6.1, 6.2]{Last}
   Let $H$ be a Hamiltonian on $\calX=\bbR^d,\bbZ^d$ with potential $V\in L^\infty(\calX)$ and $\psi$ a state whose spectral measure is not supported on a set of $\alpha$-dimensional Hausdorff measure $0$. Then, for any $m>0$, there exists a constant $C=C(\psi,m)$ with 
   \begin{align}\label{Last}
       \braket{\braket{Q_\psi ^m}}_C(T)\geq CT^{m\alpha/d}.
   \end{align}
\end{theorem}

Typically, showing the quantitative characterizations of transport beyond the scope of the Guarneri-Combes-Last theorem requires model-dependent methods. However, we note some consequences of this theorem. Since $\beta^+(2)$ is bounded from below by its time-averaged analog, in one dimension, $\mu_\varphi$ having a nonzero absolutely continuous part is enough to imply $\beta_{\varphi}^+(2)=1$. However, it is not yet known whether in one dimension $\mu_\varphi$ having nonzero absolutely continuous part is enough to ensure $\beta_{\varphi}^-(2)=1$. We note this is known to be false in higher dimensions \cite{KiselevLast,schulz1998anomalous}, where many works often seek to show averaged ballistic lower bounds precisely of the form \eqref{Last}, for example, \cite{Stolzetall} in the two dimensional almost periodic setting. This discussion motivates the following question.
\begin{question}\label{ACBT}
    Let $H$ be a one-dimensional Schr\"odinger operator with purely absolutely continuous spectrum. Do the time dynamics generated by $H$ necessarily exhibit ballistic transport in any of the (non-time averaged) senses above?
\end{question}

This question remains open even for some particularly well-studied classes of operators. For example, the following special case of Question~\ref{ACBT}, most naturally stated for Jacobi matrices:
\begin{question}\label{finitegap}
    Does strong ballistic transport hold for all states $\psi \in D(Q)$ for all finite-gap Jacobi matrices? 
\end{question}

These operators are known to have absolutely continuous spectrum. In fact, they are examples of operators in the Sodin-Yuditskii class of reflectionless Jacobi matrices whose spectrum forms a homogeneous set in the sense of Carleson \cite{SYJacobi}. So, a more general version of Question~\ref{finitegap} may be posed while still stopping short of the generality of Question~\ref{ACBT}:
\begin{question}\label{SYClass}
    Does strong ballistic transport hold for a dense set of states for Jacobi matrices in the Sodin-Yuditskii class? 
\end{question}
We note that a continuum version of this question may also be posed; see \cite{SYSchrod}.  

\par
In this context, it is worth noting that some interesting examples in higher dimensions can be found by considering a separable potential. This is based on the following proposition:
\begin{proposition}
   Let $\calH_1,\calH_2$ be two Hilbert spaces as above, where the ambient space of $\calH_i$, denoted $\calX_i$, has dimension $d_i$, for $i\in \{1,2\}$. For $i\in \{1,2\}$, let $Q_i$ be the corresponding position operator, $V_i $ a multiplication operator on $\calH_i$, and $H_i=-\Delta_i+V_i$ the corresponding Hamiltonian. We will consider the separable Hamiltonian 
   \begin{align*}
       H=H_1+H_2=-\Delta_{1+2}+V_1+V_2
   \end{align*}
   acting on $\calH_1\otimes \calH_2$, where $-\Delta_{1+2}$, is the Laplacian acting on $d_1+d_2$ dimensional space.  \par
   Let $\psi \in D(Q_1),\varphi \in D(Q_2)$. Then
   \begin{itemize}
       \item If $\psi$ exhibits ballistic transport in the strong sense for the $p$-th moment, and $\varphi$ exhibits strong ballistic transport or absence of ballistic transport for the same moment, then $\psi\otimes \varphi$ exhibits ballistic transport in the strong sense for the $p$-th moment.
       \item If $\psi $ exhibits ballistic transport in the exponent sense  for the $p$-th moment, then 
       \begin{align*}
           \beta_{\psi\otimes \varphi}^{-}(p)\geq 1.
       \end{align*}
       If it is only in the Abel averaged exponent sense, then 
       \begin{align*}
           \tilde{\beta}^{-}_{\psi\otimes \varphi}(p)\geq 1.
       \end{align*}
       \item If $\psi $ exhibits ballistic transport in the Abel average sense  for the $p$-th moment, then 
       \begin{align*}
           cT^{2p}\leq \braket{\braket{|Q|_{\psi\otimes \varphi}^{2p}}}_A.
       \end{align*} 
       
   \end{itemize}
\end{proposition}
\begin{remark}
    We note that in many applications one has upper bounds, for example for $H$ with bounded potential, such upper bounds hold trivially in the discrete setting, and by Radin-Simon \cite{RadinSimon} in the continuum. Therefore, a lower bound suffices to show ballistic transport in the appropriate sense. 
\end{remark}
\begin{remark}
    We note that since the spectral measure of a separable potential is the convolution of the spectral measures of the constituent operators, a naive application of the Guarneri-Combes-Last theorem gives a lower bound by $T^{\alpha}$ for $\alpha$ the maximum of the $\alpha$ continuity exponents of the measures, so the above direct bound improves this to $T^{2\alpha}$. 
\end{remark}
\begin{proof}
    This result follows immediately from two observations. First, we note that 
    \begin{align*}
        \|(Q_1\otimes \id)e^{-itH}(\psi \otimes \varphi)\|=\|Q_1e^{-itH_1}\psi\|\|e^{-itH_2}\varphi\|=\|Q_1e^{-itH_1}\psi\|\|\varphi\|
    \end{align*}
    and similarly to $Q_2$ and $\varphi$. This implies that any estimate in the direction of one Hamiltonian holds for $H$ along the same direction. Thus, to get strong ballistic transport in the $q_2$ direction (which is implied by strong ballistic transport), one needs to assume strong ballistic transport or the absence of ballistic transport. \par
    The second is noting that 
    \begin{align*}
        |Q|\varphi=\sqrt{|Q_1|^2+|Q_2|^2}\varphi
    \end{align*}
    and we have that $\sqrt{a^2+b^2}\geq \max\{|a|,|b|\}$.  Combining both gives the estimate for all the notions of transport.
\end{proof}

\subsection{Other Types of Transport}\label{otherTransport}

As mentioned in the introduction, in addition to ballistic transport, there are two other commonly used notions of transport: diffusive and absence of transport. In fact, one finds states with transport exponent with very different $\beta_\psi^\pm (p)$ in any dimension. \par
The absence of transport is a consequence of strong dynamical localization \cite{del1995localization}. As such, it is known to exist in some regimes (for more information, cf.\ the surveys \cite{hundertmark2008short, stolz2011introduction} and the references therein). These types of results even extend beyond the Schr\"odinger operator to long-range operators (see, for example, \cite{jitomirskaya2021upper}). In some cases, one can prove the absence of transport without localization, but in similar regimes, see, for example, \cite{damanik2007upper,  han2018quantum,  jitomirskaya2016dynamical,shamis2023upper}. In the same vein, there are models with absolutely continuous states that exhibit absence of transport. These models are infinite-dimensional (as implied by the Guarneri, Combes, and Last bounds); for more details, see \cite{vidal1999spectrum}.\par

Quantum diffusion, on the other hand, is expected to be common in higher dimensional systems in the presence of disorder \cite{stolz2011introduction}, though it is still an open question in the study of disordered systems to show that this holds. There are a few systems that have been proven to have quantum diffusion. One of these is Wegner's $N$ orbital model in the limit $N\rightarrow\infty $ (see \cite[Theorem 16]{schulz1998anomalous}). \par
We have already mentioned the result in \cite{damanik2008quantum}, where by choosing, for example, $\lambda =12$, one gets that $\beta(p)<\frac{1}{2}$, thus excluding diffusion and ballistic transport. We note that using a separable potential construction, we can extend this construction to higher dimensions (by choosing a larger value of $\lambda$). \par
We note that it is known that in some cases, $\beta^\pm (p)$ does depend on $p$ - a phenomenon called quantum intermittency; see \cite{tcheremchantsev2005dynamical} for details of one such model. Another example in which this phenomenon is known is the random polymer model, cf. \cite{MR2318858, MR1957731} for more details.

\subsection{Periodic, Quasi-Periodic, and Limit-Periodic Operators. }\label{semiperiodic}

Given the relationship described in the previous subsection and the diverse spectral properties of almost periodic operators, cf. \cite{Simon}, it will come as no surprise that the transport properties for almost periodic operators are similarly rich and, as a result, have been well-studied over the last 30 years. We give a brief overview of the literature. 

In \cite{AschKnauf,damanik2015quantum, Fillman2021}, strong ballistic transport is shown for periodic continuum Schr\"odinger operators, (block) Jacobi matrices, and Jacobi operators on $\ell^2(\mathbb Z^d)$, respectively, each of which has a purely absolutely continuous spectral type. Specifically, Asch and Knauf \cite{AschKnauf} show that for a periodic potential $V$ that is form small relative to $-\Delta$ (for example, for bounded $V$), one has strong ballistic transport in the sense of Definition~\ref{BTdef} with $W=H^1(\mathbb R^d)\cap D(Q)$. Meanwhile, in \cite{damanik2015quantum}, the corresponding notion of strong ballistic transport in the discrete setting is shown to hold. Namely, \cite{damanik2015quantum} prove strong ballistic transport for $W=D(Q)$ and $\psi \to P \psi$ a self-adjoint operator with trivial kernel. Using the common core criterion, they conclude strong resolvent convergence. With this, and ballistic upper bounds proven in \cite{damanik2010general}, they concluded, as described in Proposition~\ref{Implications}, $\beta_\varphi^\pm(p)=1$ for all $p>0$ and exponentially decaying $\varphi$. This result was subsequently extended to higher dimensional Jacobi matrices in \cite{Fillman2021}. See also recent results that examine not only ballistic transport but bounding the minimal velocity of the spread of the wave packets \cite{abdul2024slow}.\par 
Moving past the periodic setting, certain quasi-periodic operators are shown to exhibit pure point spectrum and, in fact, Anderson localization; see, for example, the works of Avila, Jitomirskaya, Bourgain, Goldstein, Schlag, Fr\"{o}hlich, Spencer, Wittwer, and Sinai \cite{AJ, Bourgain, BG, BGS, FSW, Sinai}. By Simon's result \cite{simon1990absence}, all states of these models exhibit an absence of ballistic transport. Additionally, there is \cite{bourgain2000anderson}, proving dynamical localization for more general analytic quasi-periodic potentials. We also note the results \cite{ge2019exponential,germinet2001strong,jitomirskaya2020exact}, that imply Strong Dynamical Localization for the almost Mathieu operator in certain regimes. For some additional results on pure point spectrum in the quasi-periodic class, see \cite{eliasson1997discrete, forman2021localization, ge2019exponential, jitomirskaya2018universal}.

In the limit-periodic setting, a work of Damanik and Gorodetski \cite{DamanikGorodetski} finds a dense set of discrete limit-periodic operators with pure-point spectrum. Meanwhile, P\"oschel \cite{Poschel}, as well as Damanik and Gan \cite{DamanikGan1, DamanikGan2} give examples of limit-periodic operators which exhibit an extremely strong form of Anderson localization, which in particular implies strong dynamical localization. For a discussion of this difference, we refer the reader to \cite{del1995localization}. 

However, ballistic transport, or some notion of it, is often proved for one-dimensional, almost periodic models that have purely absolutely continuous spectrum. A work by Fillman \cite{Fillman} shows that for limit-periodic Jacobi matrices that are exponentially quickly approximated (with suitably large exponent, in particular for the discrete Pastur-Tkachenko class, see \cite{Egorova}) by periodic Jacobi matrices, there is strong ballistic transport. This result is, again, leveraged to prove strong resolvent convergence with $W=\ell^1(\mathbb Z)$  in the sense of Definition~\ref{resolvent}. This result was extended to the continuous setting in \cite{young2021ballistic}, where one-dimensional limit-periodic Schr\"odinger operators in the continuum Pastur-Tkachenko class, see \cite{PasturTkachenko1}, are shown to exhibit strong ballistic transport for $W$ a set of states with suitable decay and regularity. These results motivate the following open question:
\begin{question}
    Prove strong ballistic transport for the Pastur-Tkachenko class of limit-periodic operators and a suitable set of initial states in $d\geq 2$. 
\end{question}

There have also been results in the setting of quasi-periodic operators, particularly those with ``small'' analytic potentials, which are often known to have purely absolutely continuous spectrum. In \cite{Kachkovskiy}, Kachkovskiy shows that for a large class of discrete quasi-periodic operators, Ces\`aro and phase averaged strong ballistic transport holds. More precisely, for $\theta$ denoting the phase of the sampling function defining the quasi-periodic operator $H_\theta$, and $Q_\theta(t)=e^{itH_\theta}Qe^{-itH_\theta}$, his proof shows 
\begin{align*}
    \lim\limits_{T\to\infty }\int\limits_{\mathbb T^d}\braket{Q_\theta}(T)\psi d\theta\to \int\limits_{\mathbb T^d}P(\theta)\psi d\theta
\end{align*}
for all $\psi \in D(Q)$ and $P(\theta)$ a bounded, self-adjoint operator with trivial kernel for each $\theta$ and constant norm $\|P(\theta)\|$. In particular, this implies Ces\`aro averaged strong ballistic transport along a subsequence of time scales and for almost every phase. The works of Zhao \cite{Zhao, Zhaocont}, in both the discrete and continuous setting, show that there is norm-growth ballistic transport for analytic potentials with Diophantine frequencies and small enough potentials. 

More recently, the works \cite{Kachkovskiy, Zhao,Zhaocont}, were refined by Ge and Kachkovskiy \cite{KachkovskiyGe} who proved strong ballistic transport for $W=D(Q)$ after projection onto energies for which the Schr\"odinger cocycle is reducible. Their results hold for a broad class of discrete multifrequency quasi-periodic operators.
 
There is also a result of Zhang and Zhao \cite{ZhaoZhang}, who proved ballistic transport in the exponent sense in the setting of discrete one-frequency quasi-periodic operators, remarkably only using that the operators have purely absolutely continuous spectrum. Precisely, given a one-dimensional quasi-periodic operator with phase $\theta$ and analytic potential, they are able to conclude that whenever $H_\theta$ has purely absolutely continuous spectrum for a.e. $\theta\in \mathbb T$, $\beta_\varphi^\pm(p)=1$ for $p>0$, a.e. $\theta$ and all suitably localized states $\varphi\ne 0$. This gives further evidence for the existence of an affirmative answer to Open Question~\ref{ACBT}. 

We pose the following interesting open question for a particularly important 1-dimensional quasi-periodic operator in the critical regime. 

\begin{question}
    Determine the transport exponents for $\psi\in D(Q)$ for the critical almost Mathieu operator:
    \begin{align*}
        (H_{\alpha,\theta}u)(n)=u(n+1)+u(n-1)+2\cos(2\pi(\alpha n+\theta))u(n)
    \end{align*}
    acting on $\ell^2(\bbZ)$.
\end{question}

In higher dimensions, we note the previously mentioned \cite{Stolzetall}, which proved ballistic transport in norm-growth Abel average sense, as well as \cite{black2023directional}, where strong ballistic transport on a dense set $W$ is derived as a consequence of the results there. \par

In addition to the above, there are several results concerning ballistic transport of Schr\"{o}dinger on graphs. Though it falls outside the scope of this survey, we mention Klein's work on the Bethe lattice \cite{klein1996spreading}, Aizenman and Warzel's work on tree graphs, \cite{aizenman2012absolutely}, and Korotyaev-Saburova \cite{korotyaev2017schrodinger}, who studied analogs of strip periodic potentials on more general graphs.

\section{Other Notions of Wave-Packet Spreading}\label{DiffNotion}

As mentioned in Section~\ref{intro}, ballistic transport is one way to measure the spread of a wave packet, and there are other ways to measure spread. In this review we mention escape probability (or spatial localization), scattering techniques, and dispersive estimates.
\subsection{Escape Probability}\label{Rage}
One natural way to measure the spread of wave packets is to measure the survival probability of a state from a given set. Or, in other words, to look at the quantity
\begin{align*}
    P_A(t)=\|\chi_Ae^{-itH}\psi\|,
\end{align*}
where $\chi_A $  is the indicator function of the set $A$ - which could depend on $t$ in general (though perhaps then it might be less intuitive to interpret this quantity as an escape probability). \footnote{Note that $P_A(t)$ is the square-root of the probability of finding the state in $A$ at time $t$.} \par
The most prominent result using this notion is the famous RAGE theorem, which can be stated as follows:
\begin{theorem}[RAGE\cite{AmreinGeorgescu,Ruelle}]
   Let $H$ be self-adjoint and let $A_n$ be a sequence of compact sets such that 
   $\chi_{A_n}\rightarrow 1$ pointwise, then
   \begin{align*}
       &\calH_{c}=\{\psi \in \calH\mid \lim\limits_{n\rightarrow \infty }\lim\limits_{T\rightarrow \infty }\frac{1}{T}\int\limits_0^TP_{A_n}(t)\, dt=0\},\\
       &\calH_{pp}=\{\psi \in \calH\mid \lim\limits_{n\rightarrow \infty }\sup\limits_{t\geq 0}P_{A_n^c}(t)=0\},
   \end{align*}
   where $A_n^c=\calX\setminus A_n$.  
\end{theorem}
This strong theorem says that states in the continuous spectrum are characterized by the fact their survival probability inside a compact set goes to $0$ on average, regardless of the size of the set. For states in the absolutely continuous subspace, there is the following refinement, following from the Riemann-Lebesgue lemma:
\begin{theorem}
   Let $H$ be self-adjoint and let $A$ be a compact set. Then, for $\psi\in \calH_{ac}$,
   \begin{align*}
       \lim_{t\to\infty }P_A(t)=0.
   \end{align*}
\end{theorem}
\par
In the context of this review, it is interesting to know that a state could have ballistic transport in the strong sense and still not be a purely absolutely continuous state:
\begin{proposition}\label{ppPlusBT}
    For bounded $V$, let $\psi \in \calH_{ac},\varphi \in \calH_{pp}$, be such that $\psi $ exhibits ballistic transport in the strong sense, then $\psi +\varphi$ exhibits ballistic transport in the strong sense 
\end{proposition}
\begin{proof}
    This result is immediate when considering that $\varphi \in \calH_{pp}$  has an absence of ballistic transport, as per the appendix of \cite{black2023directional}, so we get that 
    \begin{align*}
        \left\| \frac{Q(t)}{t}(\psi+\varphi)-P\psi \right\|\rightarrow 0
    \end{align*}
    as $t\rightarrow\infty $.
\end{proof}
As discussed above, the connection between the absence of ballistic transport and belonging to the pure point subspace is not that clear. 
\subsection{Scattering}\label{scattering}
Another way to measure the spread of a wave packet is by comparing it to a state with known evolution, usually by comparing it to the free evolution (perhaps up to a phase). For more detail, see \cite{RSVol3}. This is usually done by introducing the wave operators
\begin{align*}
    \Omega=\slim\limits_{t\rightarrow \infty}e^{itH}e^{-itH_0}.
\end{align*}
Then, states $\psi$ in the range of $\Omega$ are called asymptotically free, as they evolve asymptotically freely in the sense that there is a state $\phi\in \mathcal H$ whose free evolution approaches $\psi$'s under the perturbed operator: 
\begin{align*}
    \lim_{t\to\infty}\|e^{-itH_0}\phi-e^{-itH}\psi\|=0.
\end{align*}

As was noted in \cite{black2023directional}, without some assumptions on the potential, it is not clear that a state in the range of $\Omega$ exhibits ballistic transport in the strong sense- the proof there requires more refined control. In that paper, the authors found a sufficient criterion in terms similar to those of Cook's method (see, for example, \cite{RSVol3} for more details) for a state to be in the range of $\Omega$ and exhibit ballistic transport.\par
One of the main tools in scattering theory is showing the existence of certain asymptotic observables. For more details about their uses, see \cite{derezinski1997scattering}.  In this context, asymptotic velocity is proven many times using commutator techniques that require mollifying the operator. As a result, under some conditions, one gets that for any real $f\in C_c^\infty$,
\begin{align*}
    \slim\limits_{t \rightarrow\infty} f\left(\frac{Q(t)}{t}\right)=f(\Xi).
\end{align*}
This is closely related to strong resolvent convergence but is usually proven using different methods. \par
The Mourre estimate is another basic tool in scattering theory that played a major role in the proof of asymptotic completeness for $N$-body systems \cite{hunziker1999minimal}. As mentioned above, such an estimate also implies ballistic transport as well as absolutely continuous spectrum, and it plays an important part in proving asymptotic velocity estimates of the type mentioned above. 
\subsection{Dispersion}\label{disp}
Finally, another way to measure the spread of a wave packet is by using dispersive estimates.  A dispersive estimate is usually meant to be an estimate of the form
\begin{align*}
    \|e^{-itH}\psi\|_{p}\leq t^{-\alpha} \|\psi\|_{q}
\end{align*}
for some $p,q>0$, and $\alpha>0$. Most commonly, it is shown for $p=\infty$, $q=1$, since this allows interpolation with the trivial $L^2\rightarrow L^2$ bound (which is just the unitarity of the propagator). These estimates are very useful in the context of many nonlinear PDEs and their analysis (see, for example, \cite{schlag2005dispersive}). \par
Though this idea seems to be highly related to the notion of ballistic transport, given the wave packet has a constant $L^2$ norm, it doesn't imply ballistic transport. This is due to Proposition~\ref{ppPlusBT}, and noting that eigenstates don't disperse. So if $\psi$ satisfies a dispersive estimate, and $\varphi$ is an eigenstate, then we have that 
\begin{align*}
    \|e^{-itH}(\psi+\varphi)\|_p\geq \|e^{itE}\varphi\|_p-\|e^{-itH}\psi\|_p>\|\varphi\|_p-t^{-\alpha }\|\psi\|_q>\|\varphi\|_p.
\end{align*}
In particular, this norm does not decay. So, we get that these states exhibit ballistic transport but don't satisfy a dispersive estimate. 
\subsection{Discussion}
Many of the counter-examples described above are obtained by adding a pure point state to a state that transports ballistically. This leads to the following question of what can be said when there is added information about the spectral type.
\begin{question}
    Let $H$ be an operator with only continuous spectrum. Do states with ballistic transport satisfy a dispersive estimate? Do these states have to be in the absolutely continuous space?
\end{question}
Given \cite{bellissard2000subdiffusive}, one could answer this question by seeing if the subdiffusive state they generate obeys a dispersive estimate.\par
Naturally, one can also ask the following counterpart:
\begin{question}
    Let $H$ be an operator with only continuous (or even absolutely continuous) spectrum, such that all states satisfy some dispersive estimate. Do these states exhibit ballistic transport in any of the senses above?   
\end{question}

\section{Ballistic Transport in the Physics Literature}\label{physics}

\subsection{The Term ``Ballistic Transport" in the Physics Literature} \label{BTphysics}

It is interesting to note that although the term ``ballistic transport" is also used in the physics literature, it is used in a slightly different context. Ballistic transport is usually used to distinguish between two different regimes: one is dominated by collisions, and the other is dominated by free evolution. That is, the perspective of physicists is typically mesoscopic. \par
In this section, we will try to explain the source of the difference as well as explain the connection between the notions. We refer the reader to \cite{bellissard2002coherent,bellissard1995anomalous,schulz1998anomalous} for more information. In \cite{bertini2021finite}, the author gives a relatively comprehensive survey of the theory from the physics point of view for the one-dimensional system with finite temperatures. \par 
The distinction between diffusive and ballistic regimes is connected to the conductivity of a given system. The conductivity of a system is a macroscopic quantity (the inverse of the resistance), which is only well-defined if the motion of electrons is diffusive. This is an underlying assumption and can be inferred from the Einstein relation or from the assumptions used when constructing the transport theory for electrons in conductors (see \cite{ashcroft2022solid,girvin2019modern} for more details).\par 
The immediate question that arises is how to discuss conductance for a system with non-diffusive behavior (for example, with ballistic transport). This is done by adding additional mechanisms to the single-particle theory;  usually through adding thermodynamical mechanisms. For example, for a motion in a periodic system (which is ballistic), one could add a dissipation by interacting with a "bath" of free bosons that slows down the motion of electrons, and thus, on average, the motion will become diffusive. In other words, the spectral theory of a single particle is insufficient on its own to describe it.\par
Having said that, the underlying transport properties of the one-particle theory will determine some characteristics of the conductivity, especially at the limit where the temperature, $T$, goes to $0$.\par 
The common derivation of such a relationship is through the Drude formula for electric conductance, which can be written as
\begin{align}\label{Drude}
    \sigma\propto \tau,   
\end{align}
where $\sigma $ is the conductance and $\tau$ is a phenomenological parameter called ``collision time", which can be thought of as an average time between the collisions of the electron and the environment (which has a defined temperature) \cite{bellissard1995anomalous}. The precise definition is more delicate and can depend on the type of dissipation mechanism introduced. \par
Still, the underlying single particle dynamics will dictate how the conductance will depend on the collision time. For example, for ballistic transport, each collision slows down the electrons, and thus, one expects the conductance to increase with the collision time (less collision implies faster traveling). And so, the ballistic regime is described by $\sigma \propto \tau$.\par
For underlying diffusive behavior, one does not need to introduce these types of mechanisms, so the conductance is expected to be constant with the collision time. So,  the diffusive regime is described by $\sigma \propto 1$ .\par
Finally, for localized states, one needs these mechanisms to accelerate the electrons so that less collision will decrease the conductance. So the localized regime will be described as $\sigma  \propto \tau^{-1}$.  \par
Using tools from thermodynamics and many-body theory, one can find the connection between $\tau$ and the temperature, which allows one to know what type of interaction is dominant in the system (electron-electron or electron-phonon, for example). Many measurements distinguish between the regimes by measuring the temperature dependence (see, for example, \cite[Section 11.9]{girvin2019modern}). Specifically, the dependence in the limit $T\rightarrow0$ is of interest as it will imply transport expected in the idealized one-particle system. \par
Even at this point, without a clear definition, the crucial difference between the mathematical notions of ballistic transport described in Section \ref{def} and this one becomes clearer. First and foremost, the notion of ballistic transport described in thie survey doesn't take into account the thermalization processes that allow for conductance to be well-defined. 
Moreover, most of the notions we defined above (except for strong resolvent ballistic transport) are local in the sense that they measure something about a given state or a set of states; the physics definition aims to measure something about the system as a whole (or about a ``typical" wave packet in some sense). \par
It falls outside the scope of this survey to detail all the thermalization or the introduction of other mechanisms to allow for well-defined conductance. We refer the reader to \cite{bellissard2002coherent, bellissard1995anomalous, schulz1998anomalous} for more details about these derivations. In these papers, Bellissard and Schulz-Baldes use a slightly modified definition of transport exponent, which they denote by $\sigma_{\textrm{diff}}$, and which is defined by 
\begin{align*}
    2\sigma_{\textrm{diff}}=\inf \Big\{ \gamma\in \bbR : \int\limits_1^\infty \frac{\|(\braket{Q}(t)-Q)\delta_0\|^2}{t^\gamma}\frac{\,dt}{t}<\infty \Big\}.
\end{align*}
\begin{remark}
\begin{enumerate}
    \item We note that since this is done in the discrete setting, $\braket{Q}(t)-Q$ is a bounded operator for each $t$, so $\sigma_{\textrm{diff}}$ is well defined. 
    \item  This notion should be compared to the averaged ballistic transport exponents $\braket{\beta}$.
    \item When comparing those two, we note that this notion corresponds more specifically to $\braket{\beta}^+$, but it is weaker than it. This is because while $\beta^+$ gives an upper bound pointwise in time (in the form of a $\limsup$), $\sigma_{\textrm{diff}}$ requires an integral upper bound, which is of course weaker.  
    \item In their definition of the transport exponent, there is also an average over a family of Hamiltonians parametrized by $\omega$ (as one has to take the average over the hull of the Hamiltonian- generated by translations). This is common when dealing with quasicrystals; we omit this in the notation above for simplicity. 
\end{enumerate}
\end{remark}
Using this definition, they are able to show that under certain approximations, for a quasi-periodic system in the discrete setting, one can derive an ``anomalous Drude formula'',
\begin{align*}
    2\sigma_{\textrm{diff}}-1 = \inf \Big\{ \gamma\in \bbR : \int\limits_1^\infty \frac{\sigma(\tau)}{\tau^\gamma}\frac{\,d\tau}{\tau}<\infty \Big\},
\end{align*}
where $\sigma(\tau)$ is the conductance and $\tau$ has a slightly modified definition than the "collision time" defined before (this $\tau$ is the ``relaxation time'', which comes from the details of the mechanism considered in these papers).\par
This result highlights a few important differences between the mathematical perspective, which emphasizes the stronger notions of ballistic transport, and the physical perspective, which ``sees" only a relatively weak notion of ballistic transport. \par
Naturally, as we can't measure a system at $0$ temperature, a model of the measurable system should include these dissipating mechanisms. In this case, the importance of the transport exponent is only with respect to the strength of these mechanisms, so one would expect that the important quantity will be averaged across time. \par
Similar approaches to conductivity are still used in various contexts, for example, \cite{sinner2022diffusive}, where they use this asymptotic to determine the conductivity.

\subsection{Mathematical Ballistic Transport in the Physics Literature}\label{MathBTPhysics}

Though the phrase ``ballistic transport" is used in a different meaning in the physics literature, as discussed above, in some cases, notions close to the definitions given in this survey are used. We collect some of these results in this section. \par
There are several results in the context of fractal Hamiltonians, specifically in the context of one-dimensional dynamics. In \cite{abe1987fractal,wilkinson1994spectral, zhong1995quantum}, the authors considered the Fibonacci and Harper Hamiltonian, and higher dimensional Fibonacci quasilattices, and plotted the diffusion exponent, $\beta^\pm (2)$, for different ratios of frequencies. In \cite{thiem2009wave}, the authors consider a class of quasi-periodic tight-binding Hamiltonians and show numerically that it has shifted between subballistic transport (exponent $\sim 0.8$)  and subdiffusive (exponent $\sim 0.4$) behavior.  \par
It seems that there is a resurgence of interest in estimating and measuring the mean-square displacement (which is $\bar{Q}^2=(Q(t)-Q(0))^2$ ) in the context of molecular dynamics mostly by the work of Marquardt. Marquardt and Bindech, in \cite{bindech2023mean, marquardt2021mean}, considered a free particle ($V=0$), or under periodic potential, in thermal equilibrium. By thermal equilibrium, we mean that the states are the states given by weighing all possible states with a weight function of $e^{-\beta H}$, where $\beta$ is the inverse of the temperature parameter. For a system with a pure point spectrum only, this will be defined as 
\begin{align*}
    \psi =\sum_{n}e^{-\beta E_n}\varphi_n.
\end{align*}
In this setting, under some approximations, they showed that  $\bar{Q} \sim \braket{t}$, for $\braket{t}=\sqrt{t^2+1}$. In another work, Marquardt \cite{marquardt2022quantum} numerically examines this quantity for CO molecule on CU(100) substrate for short times. \par
This leads to a natural open question:
\begin{question}
    How does thermal averaging affect the known results of ballistic transport?
\end{question}
\begin{remark}
    We note that this is not trivial, as the thermal averaging ``prefers'' lower energies, where the bounds, for example, in Asch-Knauf \cite{AschKnauf}, are not uniform. So, proving that the results hold with this weight is nontrivial. In addition, one needs to make sure that such states are in the domain $D(Q)$ in order to make sense of this question.
\end{remark}

\bibliographystyle{amsplain}
\bibliography{bib}
\end{document}